%% file: multi-source_report.tex
\newif\ifreport\reporttrue
\documentclass[sigconf]{acmart}
\def\BibTeX{{\rm B\kern-.05em{\sc i\kern-.025em b}\kern-.08emT\kern-.1667em\lower.7ex\hbox{E}\kern-.125emX}}
\renewcommand\footnotetextcopyrightpermission[1]{} 
\setcopyright{none}
\usepackage{booktabs}

\usepackage{enumerate}
\usepackage{multirow}
\usepackage{bbm}
\usepackage{balance}
\usepackage{footnote}
\usepackage{tablefootnote}
\usepackage{xr}
\usepackage{xcolor}
\usepackage{subfigure}
\usepackage[lined, boxed, linesnumbered, ruled]{algorithm2e}

\theoremstyle{acmdefinition}

\theoremstyle{acmdefinition}\newtheorem{remark}{Remark}
\begin{document}






%

\title[]{Age-optimal Sampling and Transmission Scheduling in Multi-Source Systems}

\author{Ahmed M. Bedewy}
\orcid{1234-5678-9012}
\affiliation{%
  \institution{Department of ECE\\The Ohio State University}
  \city{Columbus}
  \state{OH}
}
\email{bedewy.2@osu.edu}

\author{Yin Sun}
\affiliation{%
  \institution{Department of ECE\\Auburn University}
  \city{Auburn}
  \state{AL}
}
\email{yzs0078@auburn.edu}

\author{Sastry Kompella}
\affiliation{%
  \institution{Information Technology Division\\
Naval Research Laboratory,}
 \city{Washington}
  \state{DC}
}
\email{sk@ieee.org}

\author{Ness B. Shroff}
\affiliation{%
  \institution{Departments of ECE and CSE\\The Ohio State University}
   \city{Columbus}
  \state{OH}}
\email{shroff.11@osu.edu}

\input{sections/abstract}


\maketitle




\input{sections/intro}

\input{sections/sysmodel}

\input{sections/Optimal_policy}

\input{sections/Bellman_analysis}
\input{sections/Numerical_results}

\input{sections/conclusion}

\bibliographystyle{ACM-Reference-Format}
\bibliography{MyLib}
\appendix
\input{sections/appendices_sec}

\end{document}

%% file: sections/abstract.tex
\begin{abstract}
In this paper, we consider the problem of minimizing the \emph{age of information} in a multi-source system, where samples are taken from multiple sources and sent to a destination via a channel with random delay. Due to  interference, only one source can be scheduled at a time. We consider the problem of finding a decision policy that determines the sampling times and transmission order of the sources for minimizing the total average peak age (TaPA) and the total average age (TaA) of the sources. Our investigation of this problem results in an important \emph{separation principle}: The optimal scheduling strategy and the optimal sampling strategy are independent of each other. In particular, we prove that, for any given sampling strategy, the Maximum Age First (MAF) scheduling strategy provides the best age performance among all scheduling strategies. This transforms our overall optimization problem into an optimal sampling problem, given that the decision policy follows the MAF scheduling strategy. While the zero-wait sampling strategy (in which a sample is generated once the channel becomes idle) is shown to be optimal for minimizing the TaPA, it does not always  minimize  the TaA. We use Dynamic Programming (DP) to investigate the optimal sampling problem for minimizing the TaA. Finally, we provide an approximate analysis of Bellman's equation to approximate the TaA-optimal sampling strategy by a water-filling solution which is shown to be very close to optimal through numerical evaluations.

\end{abstract}

%% file: sections/intro.tex
\section{Introduction}\label{Int}
In recent years, significant attention has been paid to the \emph{age of information} as a metric for data freshness. This is because there are a growing number of applications that require timely status updates in a variety of networked monitoring and control systems. Examples include  sensor and environment  monitoring networks, surrounding monitoring autonomous vehicles, smart grid systems, etc. The age of information, or simply age, was introduced in \cite{adelberg1995applying,cho2000synchronizing,golab2009scheduling,KaulYatesGruteser-Infocom2012}, and defined as the time elapsed since the most recently received update was generated. Unlike traditional packet-based metrics, such as throughput and delay, age is a destination-based metric that captures the information lag at the destination, and is hence more suitable in achieving the goal of timely updates. 

Early studies characterized the age in many interesting variants of queueing models \cite{KaulYatesGruteser-Infocom2012,2012ISIT-YatesKaul,2012CISS-KaulYatesGruteser,CostaCodreanuEphremides_TIT,Icc2015Pappas,KamKompellaEphremidesTIT,Bo_pull_model,2015ISITHuangModiano}, in which the update packets arrive at the queue randomly as a Poisson process. Besides these queueing theoretic studies, the work in \cite{age_optimality_multi_server,Bedewy_NBU_journal,multihop_optimal,bedewy2017age_multihop_journal} showed that Last Generated First Served (LGFS)-type policies are (nearly) optimal for minimizing any non-decreasing functional of the age process in single flow multi-server and multi-hop networks. These results hold for general system settings that include arbitrary packet generation at the source and arbitrary packet arrival times at the transmitter queue. A generalization of these results was later considered in \cite{Yin_multiple_flows} for multi-flow multi-server queueing systems, under the condition that the packet generation and arrival times are synchronized across the flows.

\begin{figure}
\includegraphics[scale=0.4]{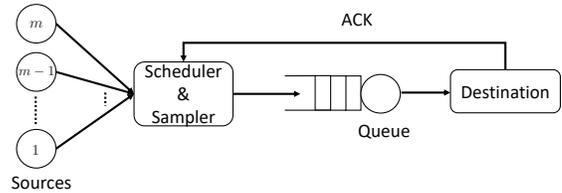}
\centering
\caption{System model}\label{Fig:sys_model}
\end{figure}
Another line of research has considered the ``generate-at-will" model \cite{BacinogCeranUysal_Biyikoglu2015ITA,2015ISITYates,SunJournal2016,sun2018sampling,Yin_mutul_info}, in which the generation times (sampling times) of the update packets are controllable. The work in \cite{SunJournal2016,sun2018sampling,Yin_mutul_info} motivated the usage of nonlinear age functions from various real-time applications and designed sampling policies for optimizing nonlinear age functions in single source systems. 
Our study here extends the work in \cite{SunJournal2016,sun2018sampling,Yin_mutul_info} by considering a multi-source information update system, as shown in Fig. \ref{Fig:sys_model}, where sources send their update packets to the destination through a channel. Due to the resource limitation (only one source can send a packet through the channel at a time), a decision maker not only controls the packet generation times, but also schedules the source transmission order. Thus, the multi-source case is more challenging.


The scheduling problem for multi-source networks with different scenarios was considered in \cite{li2013throughput,aphermedis_he2017optimal,kadota2016minimizing,kadota2016minimizing_journal,hsu2017scheduling,kadota2018optimizing,hsu2018age,yates2017status,talak2018distributed,talak2018optimizing,talak2018scheduling,talak2018optimizing2}. 
 In \cite{aphermedis_he2017optimal}, the authors found that the scheduling problem for minimizing the age in wireless networks  under physical interference constraints is NP-hard. Optimal scheduling for age minimization in a broadcast network was studied in \cite{kadota2016minimizing,kadota2016minimizing_journal,hsu2017scheduling,kadota2018optimizing,hsu2018age}, where a single source can be scheduled at a time. In contrast to our study, the generation of the update packets in \cite{li2013throughput,aphermedis_he2017optimal,kadota2016minimizing,kadota2016minimizing_journal,hsu2017scheduling,kadota2018optimizing,hsu2018age} is uncontrollable and they arrive randomly at the transmitter. Age analysis of the status updates over  a multiaccess channel was considered in \cite{yates2017status}. The studies in \cite{talak2018distributed,talak2018optimizing,talak2018scheduling,talak2018optimizing2} considered the age optimization problem in wireless network with general interference constraints and channel uncertainty. The considered sources in \cite{yates2017status,talak2018distributed,talak2018optimizing,talak2018scheduling,talak2018optimizing2} are active such that they can generate a new packet for each transmission (active sources are equivalent to zero-wait sampling strategy in our model, where a packet is generated from a source once this source is scheduled). Moreover, all the aforementioned studies for multi-source scheduling considered a time-slotted system, where a packet is transmitted in one time slot (i.e., a deterministic transmission time). Our investigation in this paper reveals that the zero-wait sampling strategy does not always minimize the age (the TaA in particular) in multi-source networks with random transmission times (which could be more than one time slot). Thus, our work here complements the studies in \cite{aphermedis_he2017optimal,kadota2016minimizing,kadota2016minimizing_journal,hsu2017scheduling,kadota2018optimizing,hsu2018age,yates2017status,talak2018distributed,talak2018optimizing,talak2018scheduling,talak2018optimizing2} by answering the following important question: What is the optimal policy that controls the packet generation times and the source scheduling to  minimize the age in a multi-source information update system with random transmission times? To that end, the main contributions of this paper are outlined as follows:
\begin{itemize}
\item We formulate the problem of finding the optimal policy that controls the sampling and scheduling strategies to minimize two age of information metrics, namely the total average peak age (TaPA) and the total average age (TaA). We show that our optimization problem has an important \emph{separation principle}: The optimal sampling strategy and the optimal scheduling strategy can be designed independently of each other. In particular, we use the stochastic ordering technique to show that, for any given sampling strategy, the Maximum Age First (MAF) scheduling strategy provides a better age performance compared to any other scheduling strategy (Proposition \ref{Thm1}). This \emph{separation principle} helps us shrink our decision policy space and transform our complicated optimization problem into an optimal sampling problem for minimizing the TaPA and TaA by fixing the scheduling strategy to the MAF strategy. 

\item We formulate the optimal sampling problem for minimizing the TaPA. We show that the zero-wait sampling strategy is the optimal one in this case (Proposition \ref{peak_zero_wait_opt}).  Hence, the MAF scheduling strategy and zero-wait sampling strategy are jointly optimal for minimizing the TaPA (Theorem \ref{thm_tapa}).  However, interestingly, we find that the zero-wait sampling strategy does not always minimize the TaA. 

\item We map the optimal sampling problem for minimizing the TaA into an equivalent optimization problem which then enables us to use Dynamic Programming (DP) to obtain the optimal sampling strategy. We show that there exists a stationary deterministic sampling strategy that can achieve optimality (Proposition \ref{thm2}). Moreover, we show that the optimal sampling strategy has a threshold property (Proposition \ref{th_thm}) that helps in reducing the complexity of the relative value iteration (RVI) algorithm (by reducing the computations required along the system state space). This results in the threshold-based sampling strategy in Algorithm \ref{alg1}. Therefore, the MAF scheduling strategy and the threshold-based sampling strategy are jointly optimal for minimizing the TaA (Theorem \ref{thm_taa}).

\item Finally, in Section \ref{Bellman}, we provide an approximate analysis of Bellman's equation whose solution is the threshold-based sampling strategy. We figure out that the water-filling solution can approximate this optimal sampling strategy. Moreover, the numerical result in Fig. \ref{avg_age_sim} shows that the performance of the water-filling solution is almost the same as that of the threshold-based sampling strategy. 
\end{itemize}
Our optimal scheduling and sampling strategies can minimize the age for any random discrete transmission times which possibly can be more than one time slot. Because of the randomness of the transmission times, our model belongs to the class of semi-Markov decision problems (SMDPs). Prior studies, such as \cite{aphermedis_he2017optimal,kadota2016minimizing,kadota2016minimizing_journal,hsu2017scheduling,kadota2018optimizing,hsu2018age,yates2017status,talak2018distributed,talak2018optimizing,talak2018scheduling,talak2018optimizing2}, considered time-slotted system. Therefore, their models belong to the class of Markov decision problems (MDPs), which cannot handle our model. Moreover, although the optimality of the MAF scheduler was shown in \cite{li2013throughput,hsu2017scheduling,kadota2016minimizing,kadota2016minimizing_journal,Yin_multiple_flows}, these studies either considered a time-slotted system \cite{li2013throughput,hsu2017scheduling,kadota2016minimizing,kadota2016minimizing_journal}, or stochastic arrivals with exponential and New-Better-than-Used (NBU) transmission times \cite{Yin_multiple_flows}. In contrast, our results are obtained by generalizing the transmission times and controlling the packet generation times. To the best of our knowledge, this is the first study that considers the joint optimization of the sampling and scheduling strategies to minimize the age in multi-source networks with random transmission times.

%% file: sections/sysmodel.tex
\section{Model and Formulation}\label{sysmod}
\subsection{Notations}
For any random variable $Z$ and an event $A$, let $\mathbb{E}[Z\vert A]$ denote the conditional expectation of $Z$ for given $A$. We use $\mathbb{N}^+$ to represent the set of non-negative integers, $\mathbb{R}^+$ is the set of non-negative real numbers, $\mathbb{R}$ is the set of real numbers, and $\mathbb{R}^n$ is the set of $n$-dimensional real Euclidean space.  We use $t^-$ to denote the time instant just before $t$.
Let $\mathbf{x}=(x_1,x_2,\ldots,x_n)$ and $\mathbf{y}=(y_1,y_2,\ldots,y_n)$ be two vectors in $\mathbb{R}^n$, then we denote $\mathbf{x}\leq\mathbf{y}$ if $x_i\leq y_i$ for $i=1,2,\ldots,n$. 
Also, we use $x_{[i]}$ to denote the $i$-th largest component of vector $\mathbf{x}$. For any bounded set $\mathcal{X}\subset\mathbb{R}$, we use $\max \mathcal{X}$ to represent the maximum of set $\mathcal{X}$, i.e., $x^*=\max \mathcal{X}$ implies that $x\leq x^*$ for all $x\in\mathcal{X}$. A set $U\subseteq \mathbb{R}^n$ is called upper if $\mathbf{y}\in U$ whenever $\mathbf{y}\geq\mathbf{x}$ and $\mathbf{x}\in U$. We will need the following definitions: 
\begin{definition} \textbf{ Univariate Stochastic Ordering:} \cite{shaked2007stochastic} Let $X$ and $Y$ be two random variables. Then, $X$ is said to be stochastically smaller than $Y$ (denoted as $X\leq_{\text{st}}Y$), if
\begin{equation*}
\begin{split}
\mathbb{P}\{X>x\}\leq \mathbb{P}\{Y>x\}, \quad \forall  x\in \mathbb{R}.
 \end{split}
\end{equation*}
\end{definition}
\begin{definition}\label{def_2} \textbf{Multivariate Stochastic Ordering:} \cite{shaked2007stochastic} 
Let $\mathbf{X}$ and $\mathbf{Y}$ be two random vectors. Then, $\mathbf{X}$ is said to be stochastically smaller than $\mathbf{Y}$ (denoted as $\mathbf{X}\leq_\text{st}\mathbf{Y}$), if
\begin{equation*}
\begin{split}
\mathbb{P}\{\mathbf{X}\in U\}\leq \mathbb{P}\{\mathbf{Y}\in U\}, \quad \text{for all upper sets} \quad U\subseteq \mathbb{R}^n.
 \end{split}
\end{equation*}
\end{definition}
\begin{definition} \textbf{ Stochastic Ordering of Stochastic Processes:} \cite{shaked2007stochastic} Let $\{X(t), t\in [0,\infty)\}$ and $\{Y(t), t\in[0,\infty)\}$ be two stochastic processes. Then, $\{X(t), t\in [0,\infty)\}$ is said to be stochastically smaller than $\{Y(t), t\in [0,\infty)\}$ (denoted by $\{X(t), t\in [0,\infty)\}\leq_\text{st}\{Y(t), t\in [0,\infty)\}$), if, for all choices of an integer $n$ and $t_1<t_2<\ldots<t_n$ in $[0,\infty)$, it holds that
\begin{align}\label{law9'}
\!\!\!(X(t_1),X(t_2),\ldots,X(t_n))\!\leq_\text{st}\!(Y(t_1),Y(t_2),\ldots,Y(t_n)),\!\!
\end{align}
where the multivariate stochastic ordering in \eqref{law9'} was defined in Definition \ref{def_2}.
\end{definition}
\subsection{System Model}
We consider a status update system with $m$ sources as shown in Fig. \ref{Fig:sys_model}, where each source observes a time-varying process. An update packet is generated from a source and is then sent over an error-free delay channel to the destination, where only one packet can be sent at a time. A decision maker (a controller) controls the generation times of the update packets and transmission order of the sources. This is known as ``generate-at-will'' model \cite{BacinogCeranUysal_Biyikoglu2015ITA,2015ISITYates,SunJournal2016} (i.e., the decision maker can generate the update packets at any time). We use $S_i$ to denote the generation time of the $i$-th generated packet, called packet $i$. Moreover, we use $r_i$ to represent the source index from which packet $i$ is generated. The channel is modeled as First-Come  First-Served  (FCFS) queue with random \emph{i.i.d.} service time $Y_i$, where $Y_i$ represents the service time of packet $i$, $Y_i\in\mathcal{Y}$, and $\mathcal{Y}\subset\mathbb{R}^+$ is a finite and bounded set. We also assume that $0<\mathbb{E}[Y_i]<\infty$ for all $i$. We suppose that the decision maker knows the idle/busy state of the server through acknowledgments (ACKs) from the destination with zero delay. To avoid unnecessary waiting time in the queue, there is no need to generate an update packet during the busy periods. Thus, a packet is served immediately once it is generated. Let $D_i$ denote the delivery time of packet $i$, where $D_i=S_i+Y_i$. After the delivery of packet $i$ at time $D_i$, the decision maker may insert a waiting time $Z_i$ before generating a new packet (hence, $S_{i+1}=D_i+Z_i$)\footnote{We suppose that $D_0=0$. Thus, we have $S_1=Z_0$.}, where $Z_i\in\mathcal{Z}$, and $\mathcal{Z}\subset\mathbb{R}^+$ is a finite and bounded set. We use $M$ to represent the  the maximum amount of waiting time allowed by the system, i.e., $M=\max\mathcal{Z}$.


\begin{figure}
\includegraphics[scale=0.25]{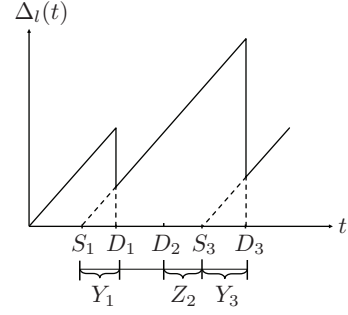}
\centering
\caption{The age $\Delta_l(t)$ of source $l$, where we suppose that $S_1, S_3\in\mathcal{G}_l$.}\label{age_proc}
\end{figure}
We use $\mathcal{G}_l$ to represent the set of generation times of the update packets that are generated from source $l$. At any time $t$, the most recently delivered packet from source $l$ is generated at time
\begin{equation}
U_l(t)=\max\{S_i\in\mathcal{G}_l : D_i\leq t\}.
\end{equation}
 The \emph{age of information}, or simply the \emph{age}, for source $l$ is defined as  \cite{adelberg1995applying,cho2000synchronizing,golab2009scheduling,KaulYatesGruteser-Infocom2012} 
\begin{equation}
\Delta_l(t)=t-U_l(t).
\end{equation} 
As shown in Fig. \ref{age_proc}, the age increases linearly with $t$ but is reset to a smaller value with the delivery of a fresher packet. We suppose that the age $\Delta_l(t)$ is right-continuous. The age process for source $l$ is given by $\{\Delta_l(t), t\geq 0\}$. We suppose that the initial age values ($\Delta_l(0^-)$ for all $l$) are known to the system.  
\subsection{Decision Policies}
A decision policy, denoted by $d$, specifies the following: i) the source scheduling strategy, denoted by $\pi$, that determines the source to be served at each transmission opportunity  $\pi\triangleq (r_1, r_2, \ldots)$, ii) the sampling strategy, denoted by $f$, that controls the packet generation times $f\triangleq (S_1, S_2, \ldots)$, or equivalently, the sequence of waiting times $f\triangleq (Z_0, Z_1, \ldots)$. Hence, $d=(\pi,f)$ implies that a decision policy $d$ follows the scheduling strategy $\pi$ and the sampling strategy $f$. Let $\mathcal{D}$ denote the set of causal decision policies in which decisions are made based on the history and current states of the system. Observe that  $\mathcal{D}$ consists of  $\Pi$ and $\mathcal{F}$, where $\Pi$ and $\mathcal{F}$ are the sets of causal scheduling and sampling strategies, respectively.  

%

After each delivery, the decision maker chooses the source to be served, and imposes a waiting time before the generation of the new packet. Next, we present our optimization problems.  
\subsection{Optimization Problem}
 We define two metrics to assess the long term age performance over our status update system in \eqref{peak_age_def} and \eqref{avg_age_def}. Consider the time interval $[0, D_n]$. For any decision policy $d=(\pi,f)$, we define the total average peak age (TaPA) as
\begin{equation}\label{peak_age_def}
\Delta_{\text{peak}}(\pi,f)=\limsup_{n\rightarrow\infty}\frac{1}{n}\mathbb{E}\left[\sum_{i=1}^{n}\Delta_{r_i}(D_i^{-})\right],
\end{equation}
and the total average age per unit time (TaA) as
\begin{equation}\label{avg_age_def}
\Delta_{\text{avg}}(\pi,f)=\limsup_{n\rightarrow\infty}\frac{\mathbb{E}\left[\sum_{l=1}^{m}\int_{0}^{D_n}\Delta_l(t)dt\right]}{\mathbb{E}\left[D_n\right]}.
\end{equation}
In this paper, we aim to minimize both the TaPA and the TaA separately. Thus, our optimization problems can be formulated as follows. We seek a decision policy $d=(\pi,f)$ that solves the following optimization problems
\begin{equation}\label{optimal_eq_p}
\bar{\Delta}_{\text{peak-opt}}\triangleq\min_{\pi \in \Pi, f\in\mathcal{F}}\Delta_{\text{peak}}(\pi,f),
\end{equation} 
and
\begin{equation}\label{optimal_eq}
\bar{\Delta}_{\text{avg-opt}}\triangleq\min_{\pi \in \Pi, f\in\mathcal{F}}\Delta_{\text{avg}}(\pi,f),
\end{equation}
where $\bar{\Delta}_{\text{peak-opt}}$ and $\bar{\Delta}_{\text{avg-opt}}$  are the optimum objective values of Problems \eqref{optimal_eq_p} and \eqref{optimal_eq}, respectively. Due to the large decision policy space, the optimization problem is quite challenging. In particular, we need to seek the optimal decision policy that controls both the sampling and scheduling strategies to minimize the TaPA and the TaA. 
On the one hand, the TaPA metric is more suitable for the applications that have an upper bound restriction on age. On the other hand, it was recently shown that, under certain conditions, information freshness measures expressed in terms of auto-correlation functions, the estimation error of signal values, and mutual information, are monotonic functions of the age \cite{sun2018sampling}. The optimal solution that we develop for minimizing TaA can be generalized to optimize the total time-average of monotonic age functions, which are useful for these applications. 
In the next section, we discuss our approach to tackle these optimization problems.

%% file: sections/Optimal_policy.tex
\section{Optimal Decision Policy}\label{optimal_policies}
We first show that our optimization problems in \eqref{optimal_eq_p} and \eqref{optimal_eq} have  an important separation principle: The scheduling strategy and the sampling strategy can be designed independently of each other. To that end, 
we show that, given the generation times of the update packets, following the Maximum Age First (MAF) scheduling strategy provides the best age performance compared to following any other scheduling strategy. What then remains to be addressed is the question of finding the best sampling strategies that solve Problems  \eqref{optimal_eq_p} and \eqref{optimal_eq}. Next, we present our approach to solve our optimization problems in detail. 
\subsection{Optimal Scheduling Strategy}
We start by defining the MAF scheduling strategy as follows:
\begin{definition}[\cite{li2013throughput,hsu2017scheduling,kadota2016minimizing,kadota2016minimizing_journal,Yin_multiple_flows}]
Maximum Age First scheduling strategy: In this scheduling strategy, the source with the maximum age is served the first among all sources. Ties are broken arbitrarily.
\end{definition}
 For simplicity, let $\pi_{\text{MAF}}$ represent the MAF scheduling strategy. The age performance resulting from following $\pi_{\text{MAF}}$ strategy is characterized as follows.
\begin{proposition}\label{Thm1}
For any given sampling strategy $f\in\mathcal{F}$, the MAF scheduling strategy minimizes the TaPA and the TaA in \eqref{peak_age_def} and \eqref{avg_age_def}, respectively, among all scheduling strategies in $\Pi$, i.e., for all $f\in\mathcal{F}$ and $\pi\in\Pi$,
\begin{equation}\label{Thm1eq1b}
\Delta_{\text{peak}}(\pi_{\text{MAF}},f)\leq \Delta_{\text{peak}}(\pi,f),
\end{equation}
\begin{equation}\label{Thm1eq2b}
\Delta_{\text{avg}}(\pi_{\text{MAF}},f)\leq \Delta_{\text{avg}}(\pi,f).
\end{equation}
\end{proposition}
\begin{proof}
One of the key ideas of the proof is as follows. Given any sampling strategy, that controls the generation times of the update packets, we only control from which source a packet is generated. We couple the policies such that the packet delivery times are fixed under all decision policies. Since we follow the MAF scheduling strategy, after each delivery, a source with maximum age becomes the source with minimum age among the $m$ sources. Under any arbitrary scheduling strategy, a packet can be generated from any source, which is not necessary the one with maximum age, and the chosen source becomes the one with minimum age among the $m$ sources after the delivery. Thus, following the MAF scheduling strategy provides a better age performance compared to following any other scheduling strategy. For details, see Appendix~\ref{Appendix_A}.
\end{proof}
Proposition \ref{Thm1} is proven by using a sample-path proof technique that was recently developed in \cite{Yin_multiple_flows}. The difference is that in \cite{Yin_multiple_flows} the authors proved the results for symmetrical packet generation and arrival processes, while  we consider here that the sampling times are controllable. 
We found that the same proof technique applies to both cases. 

\begin{remark}
The result in Proposition \ref{Thm1} can be extended to more general $\mathcal{Y}$ and $\mathcal{Z}$, i.e., $\mathcal{Y}$ and $\mathcal{Z}$ can be any uncountable sets. In other words, Proposition \ref{Thm1} holds for any arbitrary distributed service times, including continuous service times. This is because the proof of Proposition \ref{Thm1} does not depend on the service time distribution. 
\end{remark}

Proposition \ref{Thm1} 
helps us conclude the separation principle that the optimal sampling strategy can be designed independently of the optimal scheduling strategy. In particular, we are able to fix
 the scheduling strategy to the MAF strategy, and the remaining task is to search for the optimal sampling strategy. 
Hence, the optimization problems \eqref{optimal_eq_p} and \eqref{optimal_eq} reduce to the following
\begin{equation}\label{optimal_eq2'}
\bar{\Delta}_{\text{peak-opt}}\triangleq\min_{f\in\mathcal{F}}\Delta_{\text{peak}}(\pi_{\text{MAF}},f),
\end{equation}
\begin{equation}\label{optimal_eq2}
\bar{\Delta}_{\text{avg-opt}}\triangleq\min_{f\in\mathcal{F}}\Delta_{\text{avg}}(\pi_{\text{MAF}},f).
\end{equation}
Next, we seek the optimal sampling strategy for Problems \eqref{optimal_eq2'} and \eqref{optimal_eq2}. Without a confusion, we will use the term ``sampling policy" or ``sampler" to denote the sampling strategy that a decision policy can follow. Similarly, we use the term ``scheduling policy" or ``scheduler'' to denote the scheduling strategy that a decision policy can follow. 
\subsection{Optimal Sampler for Problem \eqref{optimal_eq2'}}
 \begin{figure*}[!tbp]
 \centering
 \subfigure[The age evolution of source 1.]{
  \includegraphics[scale=0.25]{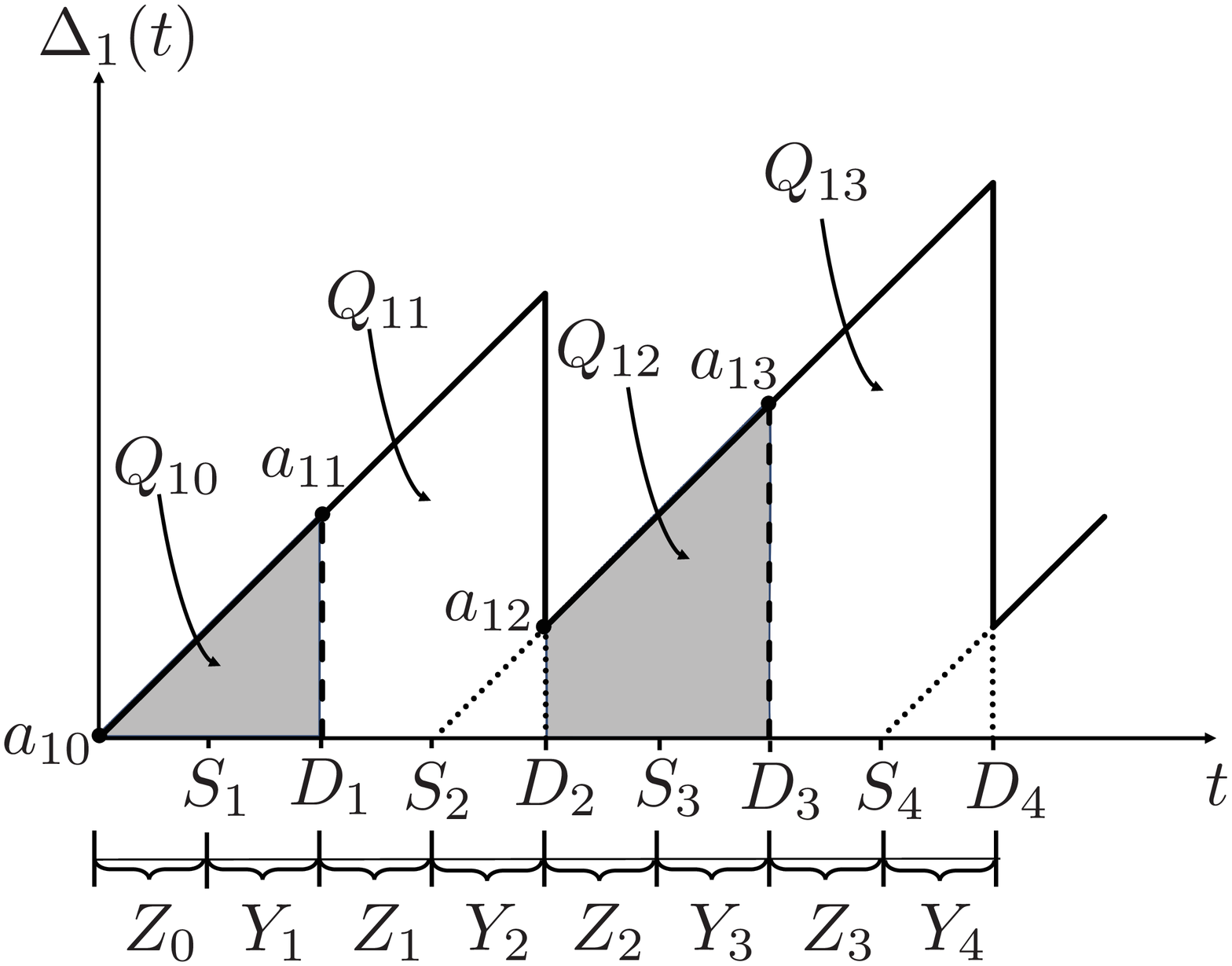}
   \label{a}
   }
 \subfigure[The age evolution of source 2.]{
  \includegraphics[scale=0.25]{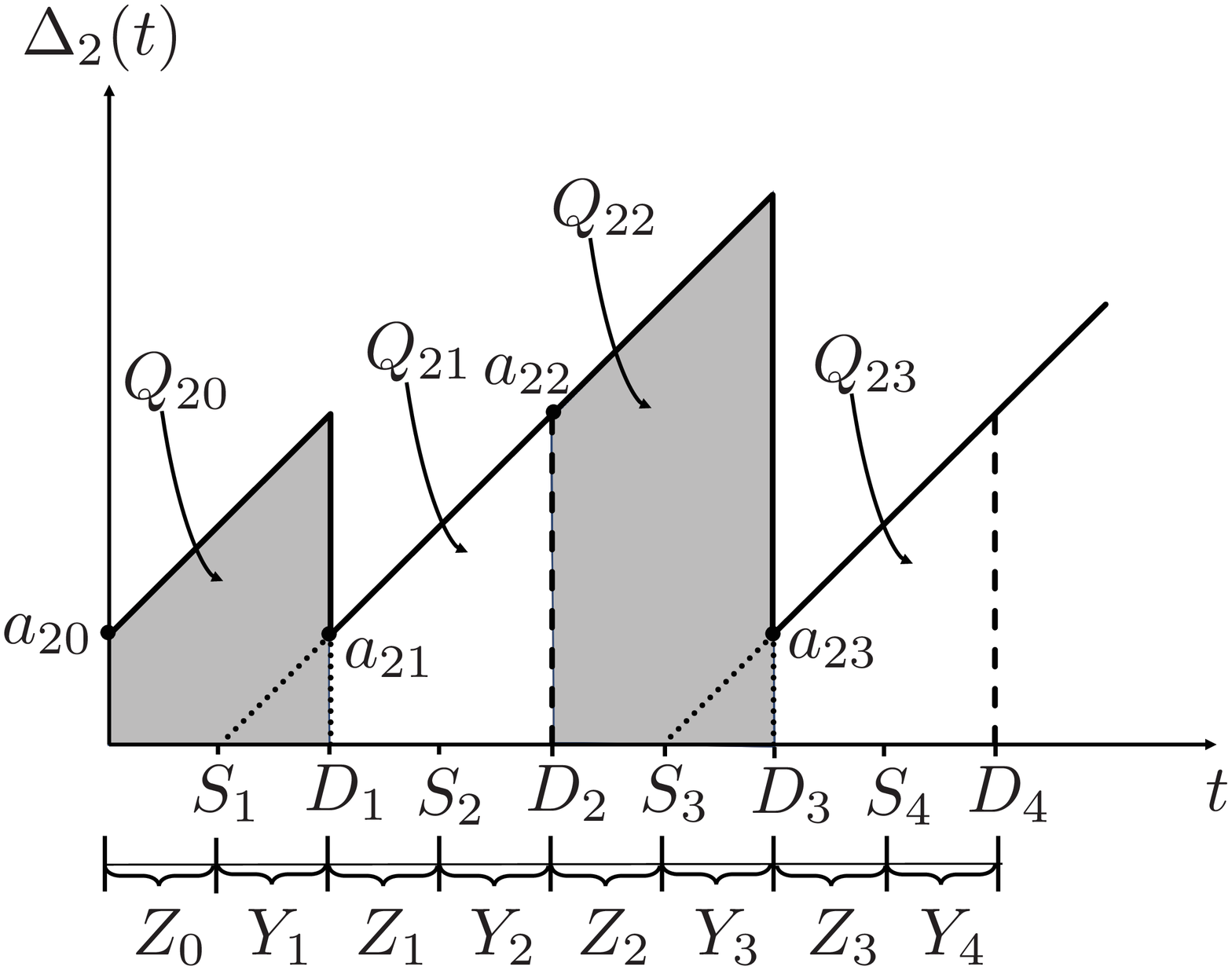}
   \label{b}
   }
   \caption{The age processes evolution of the MAF scheduler in a two-sources information update system. Source 2 has a higher initial age than Source 1. Thus, Source 2 starts service and packet 1 is generated from Source 2, which is delivered at time $D_1$. Then, Source 1 is served and packet 2 is generated from Source 1, which is delivered at time $D_2$. The same operation is repeated over time.}
   \label{fig:ages_evolv}
\end{figure*}
By fixing the scheduling policy to the MAF scheduler, the evolution of the age processes of the sources is as follows. The sampler may impose a waiting time $Z_i$ before generating packet $i+1$ at time $S_{i+1}=D_i+Z_i$ from the source with the maximum age at time $t=D_i$. Packet $i+1$ is delivered at time $D_{i+1}=S_{i+1}+Y_{i+1}$ and the age of the source with maximum age drops to the minimum age with the value of $Y_{i+1}$, while the age processes of other sources increase linearly with time without change. This operation is repeated with time and the age processes evolve accordingly. An example of age processes evolution is shown in Fig. \ref{fig:ages_evolv}. Next, we show that the zero-wait sampler minimize the TaPA.
\begin{proposition}\label{peak_zero_wait_opt}
The optimal sampler for Problem \eqref{optimal_eq2'} is the zero-wait sampler, i.e., $Z_i=0$ for all $i$. 
\end{proposition}
\ifreport 
\begin{proof}
We prove Proposition \ref{peak_zero_wait_opt} by showing that the TaPA is an increasing function of the packets waiting times $Z_i$'s. For details, see Appendix \ref{Appendix_A'}.
\end{proof}
\fi
\begin{remark}
Similar to Proposition \ref{Thm1}, the result in Proposition \ref{peak_zero_wait_opt} can be extended to more general $\mathcal{Y}$ and $\mathcal{Z}$, i.e., $\mathcal{Y}$ and $\mathcal{Z}$ can be any uncountable bounded sets.
\end{remark}
In conclusion, the optimal solution for Problem \eqref{optimal_eq_p} is manifested in the following theorem.
\begin{theorem}\label{thm_tapa}
The optimal solution for Problem \eqref{optimal_eq_p} is the MAF scheduler and the zero-wait sampler.
\end{theorem}
\begin{proof}
The theorem follows directly from Proposition \ref{Thm1} and Proposition \ref{peak_zero_wait_opt}.
\end{proof}
\subsection{Optimal Sampler for Problem \eqref{optimal_eq2}}
Although the zero-wait sampler is the optimal sampler for minimizing the TaPA, it is not clear whether it also minimizes the TaA. This is because the latter metric may not be a non-decreasing function of the waiting times as we will see later, which makes Problem \eqref{optimal_eq2} more challenging. Next, we derive the TaA when the MAF scheduler is followed and provide an equivalent mapping for Problem \eqref{optimal_eq2}.
\subsubsection{Equivalent Mapping of Problem \eqref{optimal_eq2}}
We start by deriving the TaA when the scheduling policy is fixed to the MAF scheduler. We decompose the area under each curve $\Delta_l(t)$ into a sum of disjoint geometric parts. Observing Fig. \ref{fig:ages_evolv}, this area in the time interval $[0, D_n]$, where $D_n=\sum_{i=0}^{n-1}Z_i+Y_{i+1}$, can be seen as the concatenation of the areas $Q_{li}$, $0\leq i\leq n-1$. Thus, we have
 \begin{equation}\label{integral_eq_1'}
 \int_0^{D_n}\Delta_l(t)dt=\sum_{i=0}^{n-1}Q_{li}.
 \end{equation}
Recall that we use $a_{li}$ to denote the age value for the source $l$ at time $D_i$, i.e., $a_{li}=\Delta_l(D_i)$. Then, as seen in Fig. \ref{fig:ages_evolv}, $Q_{li}$ can be expressed as 
\begin{equation}
Q_{li}=a_{li}(Z_i+Y_{i+1})+\frac{1}{2}(Z_i+Y_{i+1})^2.
\end{equation} 
Using this with \eqref{integral_eq_1'}, we get
\begin{equation}
\!\!\!\!\sum_{l=1}^{m}\int_{0}^{D_n}\Delta_l(t)dt=\sum_{i=0}^{n-1}A_{i}(Z_i+Y_{i+1})+\frac{m}{2}(Z_i+Y_{i+1})^2,
\end{equation}
where $A_{i}=\sum_{l=1}^{m}a_{li}$. The TaA can be written as
\begin{equation}\label{total_avg_age}
\!\!\!\!\limsup_{n\rightarrow\infty}\frac{\sum_{i=0}^{n-1}\mathbb{E}\left[A_{i}(Z_i+Y_{i+1})+\frac{m}{2}(Z_i+Y_{i+1})^2\right]}{\sum_{i=0}^{n-1}\mathbb{E}\left[Z_i+Y_{i+1}\right]}.
\end{equation}
Using this, the optimal sampling problem for minimizing the TaA, given that the scheduling policy is fixed to the MAF scheduler, can be formulated as
\begin{equation}\label{optimal_eq_sampler}
\!\!\!\!\bar{\Delta}_{\text{avg-opt}}\!\triangleq\!\min_{f\in\mathcal{F}}\limsup_{n\rightarrow\infty}\frac{\sum_{i=0}^{n-1}\mathbb{E}\!\left[A_{i}(Z_i\!+\!Y_{i+1})\!+\!\frac{m}{2}(Z_i\!+\!Y_{i+1})^2\right]}{\sum_{i=0}^{n-1}\mathbb{E}[Z_i\!+\!Y_{i+1}]}.\!\!
\end{equation}
Since $\mathcal{Z}$ and $\mathcal{Y}$ are bounded, $\bar{\Delta}_{\text{avg-opt}}$ is bounded as well. Note that Problem \eqref{optimal_eq_sampler} is hard to solve in the current form. Therefore, we provide an equivalent mapping for it. We consider the following optimization problem with a parameter $\beta\geq 0$:
\begin{equation}\label{equivilent_optimal_sampler}
\begin{split}
\!\!\!p(\beta)\!\triangleq\!\min_{f\in\mathcal{F}}\limsup_{n\rightarrow\infty}\frac{1}{n}\!\sum_{i=0}^{n-1}\!\mathbb{E}\!\left[(A_{i}\!-\!\beta)(Z_i\!+\!Y_{i+1})\!+\!\frac{m}{2}(Z_i\!+\!Y_{i+1})^2\!\right],\!\!
\end{split}
\end{equation}
where $p\left(\beta\right)$ is the optimal value of \eqref{equivilent_optimal_sampler}.
\begin{lemma}\label{pro_simple}
The following assertions are true:
\begin{itemize}
\item[(i).] $\bar{\Delta}_{\text{avg-opt}}\lesseqqgtr\beta$ if and only if $p(\beta)\lesseqqgtr 0$.
\item[(ii).] If $p(\beta)=0$, then the optimal sampling policies that solve \eqref{optimal_eq_sampler} and \eqref{equivilent_optimal_sampler} are identical.
\end{itemize}
\end{lemma}
\ifreport 
\begin{proof}
The proof of Lemma \ref{pro_simple} is similar to the proof of \cite[Lemma 2]{Sun_reportISIT17}. The difference is that the regenerative property of the inter-sampling times is used to prove the result in \cite{Sun_reportISIT17}; instead, we use the boundedness of the inter-sampling times to prove the result. For the sake of completeness, we modify the proof accordingly and provide it in  Appendix~\ref{Appendix_B}.
\end{proof}
\fi

As a result of Lemma \ref{pro_simple}, the solution to \eqref{optimal_eq_sampler} can be obtained by solving \eqref{equivilent_optimal_sampler} and seeking a $\beta=\bar{\Delta}_{\text{avg-opt}}\geq 0$ such that $p(\bar{\Delta}_{\text{avg-opt}})=0$. Lemma \ref{pro_simple} helps us formulate our optimization problem as a DP problem. Note that without Lemma \ref{pro_simple}, it would be quite difficult to formulate \eqref{optimal_eq_sampler} as DP problem or solve it optimally. Next, we use the DP technique to solve Problem \eqref{equivilent_optimal_sampler}.
\subsubsection{The DP problem of \eqref{equivilent_optimal_sampler}}
Following the methodology proposed in \cite{Bertsekas1996bookDPVol2}, 
when $\beta=\bar{\Delta}_{\text{avg-opt}}$, Problem \eqref{equivilent_optimal_sampler} is equivalent to an average cost per stage DP problem. According to \cite{Bertsekas1996bookDPVol2}, we describe the components of our DP problem in detail below.
\begin{itemize}
\item \textbf{States:} At stage\footnote{From henceforward, we assume that the duration of stage $i$ is $[D_i,D_{i+1})$.} $i$, the system state is specified by
\begin{equation}
\mathbf{s}(i)=(a_{[1]i}, \ldots, a_{[m]i}). 
\end{equation}
We use $\mathcal{S}$ to denote the state-space including all possible states. Notice that $\mathcal{S}$ is finite and bounded because $\mathcal{Z}$ and $\mathcal{Y}$ are finite and bounded. Also, this implies that $A_i$'s are uniformly bounded, i.e., there exists $\Lambda\in\mathbb{R}^+$ such that $A_i\leq\Lambda$ for all $i$. 

\item \textbf{Control action:} At stage $i$, the action that is taken by the sampler is $Z_i\in\mathcal{Z}$. Recall that $Z_i\leq M$ for all $i\geq 0$. 

\item \textbf{Random disturbance:} In our model, the random disturbance occurring at stage $i$ is $Y_{i+1}$, which is independent of the system state and the control action. 

\item \textbf{Transition probabilities:} If the control $Z_i=z$ is applied at stage $i$ and the service time of packet $i+1$ is $Y_{i+1}=y$, then the evolution of the system state from $\mathbf{s}(i)$ to $\mathbf{s}(i+1)$ is as follows.
\begin{equation}\label{state_evol}
\begin{split}
&a_{[m]i+1}=y,\\
&a_{[l]i+1}=a_{[l+1]i}+z+y,~ l=1,\ldots,m-1.
\end{split}
\end{equation}
We let $\mathbb{P}_{\mathbf{s}\mathbf{s}'}(z)$ denote the transition probabilities
\begin{equation*}
 \mathbb{P}_{\mathbf{s}\mathbf{s}'}(z)=\mathbb{P}(\mathbf{s}(i+1)=\mathbf{s}'\vert \mathbf{s}(i)=\mathbf{s}, Z_i=z),~\mathbf{s},\mathbf{s}'\in\mathcal{S}. 
\end{equation*}
When $\mathbf{s}=(a_{[1]},\ldots,a_{[m]})$ and $\mathbf{s}'=(a'_{[1]},\ldots,a'_{[m]})$, the law of the  transition probability is given by
\begin{equation}\label{trans_prob_eq}
\mathbb{P}_{\mathbf{s}\mathbf{s}'}(z)=\left\{ \begin{array}{cl}
\mathbb{P}(Y_{i+1}=y) & \  \ \text{if}~ a'_{[m]}=y~\text{and}\\ & \  \ a'_{[l]}=a_{[l+1]}+z+y~\text{for}~l\neq m; \\
0 & \  \text{else.} \end{array} \right. 
\end{equation}
\item\textbf{Cost Function:} Each time the system is in stage $i$ and control $Z_i$ is applied, we incur a cost 
\begin{equation}
\begin{split}
C(\mathbf{s}(i), Z_i, Y_{i+1})=&(A_{i}-\bar{\Delta}_{\text{avg-opt}})(Z_i+Y_{i+1})+\\&\frac{m}{2}(Z_i^2+2Z_iY_{i+1}+Y_{i+1}^2).
\end{split}
\end{equation}
 To simplify notation, we use the expected cost $C(\mathbf{s}(i), Z_i)$ as the cost per stage, i.e., 
\begin{equation}
C(\mathbf{s}(i), Z_i)=\mathbb{E}_{Y_{i+1}}\left[C(\mathbf{s}(i), Z_i, Y_{i+1})\right], 
\end{equation} 
where $\mathbb{E}_{Y_{i+1}}$ is the expectation with respect to $Y_{i+1}$. Hence, we have
\begin{equation}\label{cost}
\begin{split}
C(\mathbf{s}(i), Z_i)=&(A_{i}-\bar{\Delta}_{\text{avg-opt}})(Z_i+\mathbb{E}[Y])+\\&\frac{m}{2}(Z_i^2+2Z_i\mathbb{E}[Y]+\mathbb{E}\left[Y^2\right]),
\end{split}
\end{equation}
where we have used the fact that $Z_i$ and $Y_{i+1}$ are independent, and the random variable $Y$ has the same distribution as the service times $Y_i$'s. It is important to note that there exists $c\in\mathbb{R}^+$ such that $\vert C(\mathbf{s}(i), Z_i)\vert\leq c$ for all $\mathbf{s}(i)\in\mathcal{S}$ and $Z_i\in\mathcal{Z}$. This is because $\mathcal{Z}$, $\mathcal{Y}$, $\mathcal{S}$, and $\bar{\Delta}_{\text{avg-opt}}$ are bounded. 
\end{itemize}
In general, the average cost per stage under a sampling policy $f\in\mathcal{F}$ is given by 
\begin{equation}\label{avg_cost_per_stage_1'}
\limsup_{n\rightarrow\infty}\frac{1}{n}\mathbb{E}\left[ \sum_{i=0}^{n-1}C(\mathbf{s}(i), Z_i)\right].
\end{equation}
We say that a sampling policy $f\in\mathcal{F}$ is \emph{average-optimal} if it minimizes the average cost per stage in \eqref{avg_cost_per_stage_1'}.
Our objective is to find the average-optimal sampling policy. A policy $f$ is called history-dependent if the control $Z_i$ depends on the entire past history, i.e., it depends on $\mathbf{s}(0),\ldots,\mathbf{s}(i)$ and $Z_0,\ldots,Z_{i-1}$. A policy is stationary if $Z_i=Z_j$ whenever $\mathbf{s}(i)=\mathbf{s}(j)$ for any $i$, $j$. In addition, a randomized policy assigns a probability distribution over the control set such that it chooses a control randomly according to this distribution, while a deterministic policy selects an action with certainty.
According to \cite{Bertsekas1996bookDPVol2}, there may not exist a stationary deterministic policy that is average-optimal. In the next proposition, we show that there is actually a stationary deterministic policy that is average-optimal to our problem. 
\begin{proposition}\label{thm2}
There exist a scalar $\lambda$ and a function $h$ that satisfy the following Bellman's equation
\begin{equation}\label{bell1'}
\lambda+h(\mathbf{s})=\min_{z\in\mathcal{Z}}\left[ C(\mathbf{s},z)+\sum_{\mathbf{s}'\in\mathcal{S}}\mathbb{P}_{\mathbf{s}\mathbf{s}'}(z)h(\mathbf{s}')\right],
\end{equation}
where $\lambda$ is the optimal average cost per stage that is independent of the initial state $\mathbf{s}(0)$ and satisfies 
\begin{equation}
\lambda=\lim_{\alpha\rightarrow 1}(1-\alpha)J_{\alpha}(\mathbf{s}), \forall \mathbf{s}\in\mathcal{S}, 
\end{equation}
and $h(\mathbf{s})$ is the relative cost function that, for any state $\mathbf{o}$, satisfies
\begin{equation}\label{relative_cost_eq}
h(\mathbf{s})=\lim_{\alpha\rightarrow 1}(J_\alpha(\mathbf{s})-J_\alpha(\mathbf{o})), \forall \mathbf{s}\in\mathcal{S},
\end{equation}
where $J_\alpha(\mathbf{s})$ is the optimal total expected $\alpha$-discounted cost function, which is defined by
\begin{equation}
 J_\alpha(\mathbf{s})=\min_{f\in\mathcal{F}}\limsup_{n\rightarrow\infty}\mathbb{E}\left[ \sum_{i=0}^{n-1}\alpha^iC(\mathbf{s}(i), Z_i)\right],\mathbf{s}(0)=\mathbf{s}\in\mathcal{S},
\end{equation} 
where $0<\alpha<1$ is the discount factor. Furthermore, there exists a stationary deterministic policy that attains the minimum in \eqref{bell1'} for each $\mathbf{s}\in\mathcal{S}$ and is average-optimal.
\end{proposition}  
\ifreport 
\begin{proof}
The proof idea of this proposition is different from those used in literature such as \cite{hsu2017scheduling,hsu2018age}. In particular, we show that  for every two states $\mathbf{s}$ and $\mathbf{s}'$, there exists a stationary deterministic policy $f$ such that for some $k$, we have $\mathbb{P}\left[\mathbf{s}(k)=\mathbf{s}'\vert\mathbf{s}(0)=\mathbf{s}, f\right] > 0$, i.e., we have a communicating Markov decision process (MDP). For details, see Appendix \ref{Appendix_C}.
\end{proof}
\fi
We can deduce from Proposition \ref{thm2} that the optimal waiting time is a fixed function of the state $\mathbf{s}$. Next, we use the relative value iteration algorithm to obtain the optimal sampler for minimizing the TaA, and then exploit the structure of our problem to reduce its complexity.



\subsubsection{Optimal Sampler Structure}
The relative value iteration (RVI) algorithm \cite[Section 9.5.3]{puterman2005markov}, \cite[Page 171]{kaelbling1996recent} can be used to solve the Bellman's equation \eqref{bell1'}. Starting with an arbitrary state $\mathbf{o}$, a single iteration for the RVI algorithm is given as follows:
\begin{equation}\label{RVI1}
\begin{split}
&Q_{n+1}(\mathbf{s},z)= C(\mathbf{s},z)+\sum_{\mathbf{s}'\in\mathcal{S}}\mathbb{P}_{\mathbf{s}\mathbf{s}'}(z)h_n(\mathbf{s}'),\\
&J_{n+1}(\mathbf{s})=\min_{z\in\mathcal{Z}}(Q_{n+1}(\mathbf{s},z)),\\
&h_{n+1}(\mathbf{s})= J_{n+1}(\mathbf{s})-J_{n+1}(\mathbf{o}),
\end{split}
\end{equation}
where $Q_{n+1}(\mathbf{s},z)$, $J_{n}(\mathbf{s})$, and $h_{n}(\mathbf{s})$ denote the state action value function, value function, and relative value function for iteration $n$, respectively. At the beginning, we set $J_{0}(\mathbf{s})=0$ for all $\mathbf{s}\in\mathcal{S}$, and then we repeat the iteration of the RVI algorithm as described before 
\footnote{ According to \cite{puterman2005markov, kaelbling1996recent}, a sufficient condition for the convergence of the RVI algorithm is the aperiodicity of the transition matrices of stationary deterministic optimal policies. In our case, these transition matrices depend on the service times. This condition can always be achieved by applying the aperiodicity transformation as explained in \cite[Section 8.5.4]{puterman2005markov}, which is a simple transformation. However, This is not always necessary to be done.}.
 
 The complexity of the RVI algorithm is high due to many sources (i.e., curse of dimensionatlity \cite{powell2007approximate}). Thus, we need to simplify the RVI algorithm. To that end, we show that the optimal sampler has a threshold property that can reduce the complexity of the RVI algorithm. 
\begin{proposition}\label{th_thm}
Let $A_{s}$ be the sum of the age values of state $\mathbf{s}$. Then, the optimal waiting time of any state $\mathbf{s}$ with $A_{s}\geq (\bar{\Delta}_{\text{avg-opt}}-m\mathbb{E}[Y])$ is zero. 
\end{proposition}
\ifreport 
\begin{proof}
See Appendix \ref{Appendix_E}.
\end{proof}
\fi
\begin{algorithm}[!t]
\footnotesize
\SetKwData{NULL}{NULL}
\SetCommentSty{footnotesize} 
\textbf{given} $l=0$, sufficiently large $u$, tolerance $\epsilon_1>0$, tolerance $\epsilon_2>0$\;
\While{$u-l>\epsilon_1$}{
$\beta=\frac{l+u}{2}$\;
$J(\mathbf{s})=0$, $h(\mathbf{s})=0$, $h_{\text{last}}(\mathbf{s})=0$ for all states $\mathbf{s}\in\mathcal{S}$\;
\While{$\max_{\mathbf{s}\in\mathcal{S}}\vert h(\mathbf{s})-h_{\text{last}}(\mathbf{s})\vert> \epsilon_2$} {
\For{\textbf{each} $\mathbf{s}\in\mathcal{S}$}{
\uIf{$A_s\geq (\beta-m\mathbb{E}[Y])$}{
$z^*_s=0$\;}
\Else{
$z^*_s=\text{argmin}_{z\in\mathcal{Z}}(A_{s}-\beta)(z+\mathbb{E}[Y])+\frac{m}{2}(z^2+2z\mathbb{E}[Y]+\mathbb{E}\left[Y^2\right])+\sum_{\mathbf{s}'\in\mathcal{S}}\mathbb{P}_{\mathbf{s}\mathbf{s}'}(z)h(\mathbf{s}')$\;
}
$J(\mathbf{s})=(A_{s}-\beta)(z+\mathbb{E}[Y])+\frac{m}{2}(z^2+2z\mathbb{E}[Y]+\mathbb{E}\left[Y^2\right])+\sum_{\mathbf{s}'\in\mathcal{S}}\mathbb{P}_{\mathbf{s}\mathbf{s}'}(z^*_s)h(\mathbf{s}')$\;
}
$h_{\text{last}}(\mathbf{s})=h(\mathbf{s})$\;
$h(\mathbf{s})=J(\mathbf{s})-J(\mathbf{o})$\;
}
\uIf{$J(\mathbf{o})\geq 0$}{
$u=\beta$\;}
\Else{
$l=\beta$\;
}
}
\caption{\small Threshold-based sampler based on RVI algorithm.}\label{alg1}
\end{algorithm}
We can exploit Proposition \ref{th_thm} to reduce the complexity of the RVI algorithm as follows. The optimal waiting time for any state $\mathbf{s}$ whose $A_{s}\geq (\bar{\Delta}_{\text{avg-opt}}-m\mathbb{E}[Y])$ is zero. Thus, we need to solve \eqref{RVI1} only for the states whose $A_{s}\leq (\bar{\Delta}_{\text{avg-opt}}-m\mathbb{E}[Y])$. As a result, we reduce the number of computations required along the system state space, which reduce the complexity of the RVI algorithm. Note that $\bar{\Delta}_{\text{avg-opt}}$ can be obtained using the bisection method or any other one-dimensional search methods. Combining this with the result of Proposition \ref{th_thm} and the RVI algorithm,  we propose the ``threshold-based sampler'' in Algorithm \ref{alg1}, where $z^*_s$ is the optimal waiting time for state $\mathbf{s}$. In the outer  layer of Algorithm \ref{alg1},  bisection is employed to obtain $\bar{\Delta}_{\text{avg-opt}}$, where $\beta$ converges to $\bar{\Delta}_{\text{avg-opt}}$.

Note that, according to  \cite{puterman2005markov,kaelbling1996recent}, $J(\mathbf{o})$ in Algorithm \ref{alg1} converges to the optimal average cost per stage. Moreover, the value of $u$ in Algorithm \ref{alg1} can be initialized to the value of the TaA of the zero-wait sampler (as the TaA of the zero-wait sampler provides an upper bound on the optimal TaA), which can be easily calculated. 

The RVI algorithm and Whittle's methodology have been used in literature to obtain the optimal age scheduler in a time-slotted multi-source networks (e.g.,\cite{hsu2017scheduling,hsu2018age}). Since they considered time-slotted system, their model belongs to the class of MDPs. In contrast, we consider random discrete transmission times that can be more than one time slot. Thus, our model belongs to the class of SMDPs, and hence is different from those in \cite{hsu2017scheduling,hsu2018age}.

In conclusion, the optimal solution for Problem \eqref{optimal_eq} is manifested in the following theorem.
\begin{theorem}\label{thm_taa}
The optimal solution for Problem \eqref{optimal_eq} is the MAF scheduler and the threshold-based sampler. 
\end{theorem}
\begin{proof}
The theorem follows directly from Proposition \ref{Thm1}, Proposition \ref{thm2}, and Proposition \ref{th_thm}
\end{proof}

Although the work in \cite{SunJournal2016} provided the solution of the optimal sampling problem for minimizing the age in single source systems, its results hold only when there is a bound on the waiting times. In this paper, we show that we can indeed generalize our results and eliminate the upper bound on the waiting times, $M$. In particular, we show that for a large enough $M$, the obtained solution is as if the upper bound $M$ is removed. Let $\bar{\Delta}^{\infty}_{\text{avg-opt}}$ and  $f^*_{\infty}$ denote the optimal TaA and optimal sampler  when the upper bound on the waiting times is $\infty$, respectively. Moreover, let $\bar{\Delta}^{M}_{\text{avg-opt}}$ and $f^*_{M}$ denote the optimal TaA and optimal sampler when the upper bound on the waiting times is $M$, respectively. Our result is manifested as follows.
\begin{theorem}\label{upperbound_removal}
 There exists $N_o\in\mathbb{R}^+$ such that for all $M\geq N_o$, we have
\begin{equation}\label{opt_samp_inf}
f^*_{\infty}=f^*_{M}, ~\bar{\Delta}^{\infty}_{\text{avg-opt}}=\bar{\Delta}^{M}_{\text{avg-opt}}.
\end{equation}
%
\end{theorem}
\ifreport 
\begin{proof}
See Appendix \ref{Appendix_D}.
\end{proof}
\fi
The authors in \cite{hsu2017scheduling} obtained a similar result, where they showed that the truncated MDP for solving the scheduling problem in multi-source systems with stochastic arrivals converges to the original MDP (with infinite state space) as the truncate level goes to $\infty$. However, for very low arrival rates, the truncate level can be achieved even under optimal policies. In this paper, we show that the truncate level on the waiting times can actually be removed without affecting the optimal result.

%% file: sections/Bellman_analysis.tex
\section{Bellman's Equation Approximation}\label{Bellman}
In this section, we provide an approximate analysis for Bellman equation in \eqref{bell1'} in order to find a simple algorithm to solve Problem \eqref{equivilent_optimal_sampler}. For a given state $\mathbf{s}$, we denote the next state given $z$ and $y$ by $\mathbf{s}'(z,y)$. From the state evolution in \eqref{state_evol} and the transition probability equation \eqref{trans_prob_eq}, Bellman's equation in \eqref{bell1'} can be rewritten as 
\begin{equation}\label{bell2'}
\!\!\!\!\!\!\!\lambda=\min_{z}\left[ C(\mathbf{s},z)+\sum_{y\in\mathcal{Y}}\mathbb{P}(Y=y)(h(\mathbf{s}'(z,y))-h(\mathbf{s}))\right].
\end{equation}
Although $h(\mathbf{s})$ is discrete, we can interpolate the value of $h(\mathbf{s})$ between the discrete values so that it is differentiable by following the same approach in \cite{bettesh2006optimal} and \cite{wangdelay}. Let $\mathbf{s}=(a_{[1]},\ldots,a_{[m]})$, then using the first order Taylor approximation around a state $\mathbf{t}=(a_{[1]}^t,\ldots,a_{[m]}^t)$ (some fixed state), we get
\begin{equation}\label{taylor_exp1}
h(\mathbf{s})\approx h(\mathbf{t})+\sum_{l=1}^{m}(a_{[l]}-a_{[l]}^t)\frac{\partial h(\mathbf{t})}{\partial a_{[l]}}.
\end{equation}
Again, we use the first order Taylor approximation around the state $\mathbf{t}$, together with the state evolution in \eqref{state_evol}, to get
\begin{equation}\label{taylor_exp2}
\!\!h(\mathbf{s}'(z,y))\!\!\approx\!\!h(\mathbf{t})\!+\!(y-a_{[m]}^t)\frac{\partial h(\mathbf{t})}{\partial a_{[m]}}\!+\!\sum_{l=1}^{m-1}(a_{[l+1]}-a_{[l]}^t+z+y)\frac{\partial h(\mathbf{t})}{\partial a_{[l]}}.\!\!
\end{equation}
From \eqref{taylor_exp1} and \eqref{taylor_exp2}, we get
\begin{equation}
\begin{split}
h(\mathbf{s}'(z,y))\!\!-\!\!h(\mathbf{s})\!\!\approx\!\! (y\!\!-\!\!a_{[m]})\frac{\partial h(\mathbf{t})}{\partial a_{[m]}}\!\!+\!\!\sum_{l=1}^{m-1}(a_{[l+1]}\!\!-\!\!a_{[l]}\!\!+\!\!z\!\!+\!\!y)\frac{\partial h(\mathbf{t})}{\partial a_{[l]}}.
\end{split}
\end{equation}
This implies that
 \begin{equation}\label{taylor1}
\begin{split}
\!\!\!\!\!\!\sum_{y\in\mathcal{Y}}\!\mathbb{P}(Y\!=\!y)(h(\mathbf{s}'(z,y))\!-\!h(\mathbf{s}))\!\approx&(\mathbb{E}[Y]\!-\!a_{[m]})\frac{\partial h(\mathbf{t})}{\partial a_{[m]}}+\\&\sum_{l=1}^{m-1}\!(a_{[l+1]}\!-\!a_{[l]}\!+\!z\!+\!\mathbb{E}[Y])\frac{\partial h(\mathbf{t})}{\partial a_{[l]}}.\!\!\!\!
\end{split}
\end{equation}
Using \eqref{bell2'} with \eqref{taylor1}, we can get the following approximated Bellman's equation. 
\begin{equation*}\label{approx1}
\begin{split}
\lambda\approx&\min_z(A_s-\bar{\Delta}_{\text{avg-opt}})(z+\mathbb{E}[Y])+\frac{m}{2}((z)^2+2z\mathbb{E}[Y]+\mathbb{E}\left[Y^2\right])\!\!\!\!\!\!\!\!\!\!\!\!\!\!\!\!\!\!\!\!\!\!\\&+(\mathbb{E}[Y]-a_{[m]})\frac{\partial h(\mathbf{t})}{\partial a_{[m]}}+\sum_{l=1}^{m-1}(a_{[l+1]}-a_{[l]}+z+\mathbb{E}[Y])\frac{\partial h(\mathbf{t})}{\partial a_{[l]}},
\end{split}
\end{equation*}
where $A_s$ is the sum of the age values of state $\mathbf{s}$. A necessary condition for minimizing the RHS of the previous equation is to set its derivative to zero. We get
\begin{equation}\label{opt}
A_s-\bar{\Delta}_{\text{avg-opt}}+mz+m\mathbb{E}[Y]+\sum_{l=1}^{m-1}\frac{\partial h(\mathbf{t})}{\partial a_{[l]}}=0.
\end{equation}
Rearranging \eqref{opt}, we get
\begin{equation}\label{optimal_sol1}
\hat{z}^*_s=\left[\frac{\bar{\Delta}_{\text{avg-opt}}-m\mathbb{E}[Y]-\sum_{l=1}^{m-1}\frac{\partial h(\mathbf{t})}{\partial a_{[l]}}}{m}-\frac{A_s}{m} \right]^+,
\end{equation}
where $\hat{z}^*_s$ is the optimal solution of the approximated Bellman's equation for state $\mathbf{s}$. Note that the term  $\sum_{i=1}^{m-1}\frac{\partial h(\mathbf{t})}{\partial a_{[i]}}$ is constant. Hence, \eqref{optimal_sol1} can be written as
\begin{equation}\label{fin_approx_sol}
\hat{z}^*_s=\left[th-\frac{A_s}{m} \right]^+,
\end{equation}
where we have used Theorem \ref{upperbound_removal} to eliminate the upper bound $M$ in \eqref{fin_approx_sol} (or simply $M$ can be set to be large enough such that it is greater than the optimal threshold in \eqref{fin_approx_sol}). The solution in \eqref{fin_approx_sol} is in the form of the water-filling solution as we compare a fixed threshold ($th$) with the average age of a state $\mathbf{s}$. The solution in \eqref{fin_approx_sol} suggests that the water-filling solution can approximate the optimal solution of the original Bellman's equation in \eqref{bell1'}. The optimal threshold ($th$) in \eqref{fin_approx_sol} can be obtained using a golden-section method \cite{press1992golden}. We evaluate the performance of the water-filling solution obtained in \eqref{fin_approx_sol} in the next section.

%% file: sections/Numerical_results.tex
\section{Numerical Results}\label{Simulations}
We present some numerical results to verify our theoretical results. We consider an information update system with $m=3$ sources. The packet transmission times are either 0 or 3 with probability $p$ and $1-p$, respectively. We use ``RAND" to represent the random scheduler in which a source is chosen randomly to be served. Also, we use ``constant-wait sampler" to represent the sampler that imposes a constant waiting time after each delivery with $Z_i=0.3\mathbb{E}[Y],~\forall i$. 

\begin{figure}[t]
\centering
\includegraphics[scale=0.4]{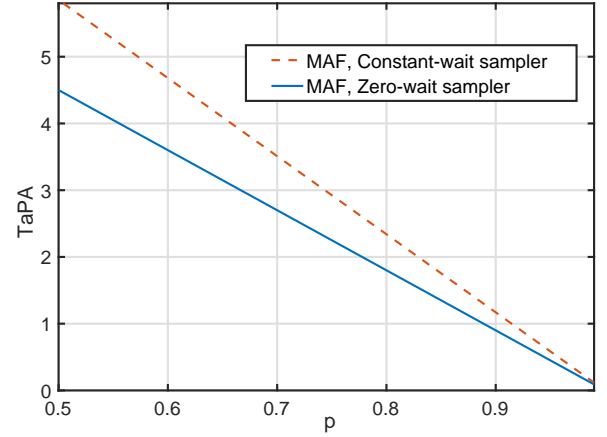}
\caption{TaPA versus transmission probability $p$ for an update system with $m=3$ sources.}
\label{peak_age_sim}
\end{figure}
Fig. \ref{peak_age_sim} illustrates the TaPA versus the transmission probability $p$. As we can observe, with fixing the scheduling policy to the MAF scheduler, the zero-wait sampler provides a lower TaPA compared to the constant-wait sampler. This observation agrees with Proposition \ref{peak_zero_wait_opt}. However, as we will see later, zero-wait sampler does not always minimize the TaA. 

\begin{figure}[t]
\centering
\includegraphics[scale=0.4]{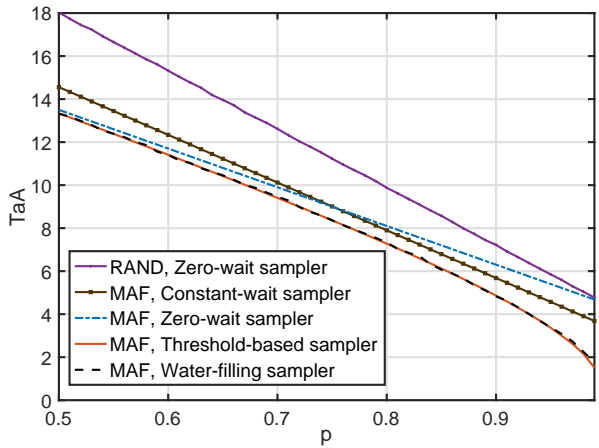}
\caption{TaA versus transmission probability $p$ for an update system with $m=3$ sources.}
\label{avg_age_sim}
\end{figure}
Fig. \ref{avg_age_sim} illustrates the TaA versus the transmission probability $p$.  For the zero-wait sampler, we find that following the MAF scheduler provides a lower TaA than that is resulting from following the RAND scheduler. This agrees with Proposition \ref{Thm1}. Moreover, when the scheduling policy is fixed to the MAF scheduler, we find that the TaA results from the threshold-based sampler is lower than those result from the zero-wait sampler and the constant-wait sampler. This observation implies the following: i) The zero-wait sampler does not necessarily minimize the TaA, ii) optimizing the scheduling policy only is not enough to minimize the TaA and we have to optimize both the scheduling policy and the sampling policy together to minimize the TaA. 
 Finally, as we can observe, the TaA resulting from the water-filling sampler almost coincides on the TaA resulting from the threshold-based sampler. 

%% file: sections/conclusion.tex
\section{Conclusion}\label{Concl}
In this paper, we studied the problem of finding the optimal decision policy that controls the packet generation times and transmission order of the sources to minimize the TaPA and TaA in multi-source information update system. We showed the MAF scheduler and the zero-wait sampler are jointly optimal for minimizing the TaPA. Moreover, we showed that the MAF scheduler and the threshold-based sampler, that is based on the RVI algorithm, are jointly optimal for minimizing the TaA. Finally, we provided an approximate analysis of Bellman's equation and showed that the water-filling solution can approximate the threshold-based sampler. The numerical result showed that the performance of the water-filling solution is almost the same as that of the threshold-based sampler.

\begin{acks}
This work has been supported in part by ONR grants N00014-17-1-2417 and N00014-15-1-2166, Army Research Office grants W911NF-14-1-0368 and MURI W911NF-12-1-0385,  NSF grants CNS-1446582, CNS-1421576, CNS-1518829, and CCF-1813050, and a grant from the Defense Thrust Reduction Agency HDTRA1-14-1-0058.
\end{acks}

%% file: sections/appendices_sec.tex
\section{Proof of Proposition \ref{Thm1}}\label{Appendix_A}
Let the vector $\mathbf{\Delta}_\pi(t)=(\Delta_{[1],\pi}(t),\ldots, \Delta_{[m],\pi}(t))$ denote the system
state at time $t$ when the scheduling strategy $\pi$ is followed,  where $\Delta_{[l],\pi}(t)$ is the $l$-th largest age of the sources at time $t$ under the scheduling strategy $\pi$. Let $\{\mathbf{\Delta}_\pi(t), t\geq 0\}$ denote the state process when the scheduling strategy $\pi$ is followed. 
For notational simplicity, let  $P$ represent the MAF scheduling strategy. Throughout the proof, we assume that $\mathbf{\Delta}_\pi(0^-)=\mathbf{\Delta}_P(0^-)$ for all $\pi$ and the sampling strategy is fixed to an arbitrarily chosen one. The key step in the proof of Proposition \ref{Thm1} is the following lemma, where we compare the scheduling strategy $P$ with any arbitrary scheduling strategy $\pi$.

\begin{lemma}\label{lem1thm1}
Suppose that $\mathbf{\Delta}_\pi(0^-)=\mathbf{\Delta}_P(0^-)$ for all scheduling strategy $\pi$ and the sampling strategy is fixed, then we have 
\begin{equation}\label{lem1thm1eq1}
\{\mathbf{\Delta}_P(t), t\geq 0\}\leq_{\text{st}}\{\mathbf{\Delta}_\pi(t), t\geq 0\}
\end{equation}
\end{lemma}
We use a coupling and forward induction to prove Lemma \ref{lem1thm1}. For any scheduling strategy $\pi$, suppose that the stochastic processes $\widetilde{\mathbf{\Delta}}_P(t)$ and $\widetilde{\mathbf{\Delta}}_\pi(t)$ have the same stochastic laws as $\mathbf{\Delta}_P(t)$ and $\mathbf{\Delta}_\pi(t)$. The state processes $\widetilde{\mathbf{\Delta}}_P(t)$ and $\widetilde{\mathbf{\Delta}}_\pi(t)$ are coupled such that the packet service times are equal under both scheduling policies, i.e., $Y_i$'s are the same under both scheduling policies. Such a coupling is valid since the service time distribution is fixed under all policies. Since the sampling strategy is fixed, such a coupling implies that the packet generation and delivery times are the same under both scheduling strategies. According to Theorem 6.B.30 of \cite{shaked2007stochastic}, if we can show 
\begin{equation}\label{Thm1eq1}
\mathbb{P}\left[\widetilde{\mathbf{\Delta}}_P(t)\leq\widetilde{\mathbf{\Delta}}_\pi(t), t\geq 0\right]=1,
\end{equation}
then \eqref{lem1thm1eq1} is proven. To ease the notational burden, we will omit the tildes on the coupled versions in this proof and just use $\mathbf{\Delta}_P(t)$ and $\mathbf{\Delta}_\pi(t)$. Next, we compare strategy $P$ and strategy $\pi$ on a sample path and prove \eqref{lem1thm1eq1} using the following lemma:
\begin{lemma}[Inductive Comparison]\label{lem2thm1}
Suppose that a packet with generation time $S$ is delivered under the scheduling strategy $P$ and the scheduling strategy $\pi$ at the same time $t$. The system state of the scheduling strategy $P$ is $\mathbf{\Delta}_P$ before the packet delivery, which becomes $\mathbf{\Delta}'_P$ after the packet delivery. The system state of the scheduling strategy $\pi$ is $\mathbf{\Delta}_\pi$ before the packet delivery, which becomes $\mathbf{\Delta}'_\pi$ after the packet delivery. If
\begin{equation}\label{lem2thm1eq1}
\Delta_{[i],P}\leq \Delta_{[i],\pi}, i=1,\ldots,m,
\end{equation}
then
\begin{equation}\label{lem2thm1eq2}
\Delta_{[i],P}'\leq \Delta_{[i],\pi}', i=1,\ldots,m.
\end{equation}
\end{lemma}
Lemma \ref{lem2thm1}  is proven by following the proof idea of \cite[Lemma 2]{Yin_multiple_flows}. For the sake of completeness, we provide the proof of Lemma \ref{lem2thm1}  as follows. 
\begin{proof}
Since only one source can be scheduled at a time and the scheduling strategy $P$ is the MAF scheduling strategy, the packet with generation time $S$ must be generated from the source with maximum age $\Delta_{[1],P}$, call it source $l^*$. In other words, the age of source $l^*$ is reduced from the maximum age $\Delta_{[1],P}$ to the minimum age $\Delta_{[m],P}'=t-S$, and the age of the other $(m-1)$ sources	remain unchanged. Hence, 
\begin{equation}\label{prlem2thm1eq1}
\begin{split}
&\Delta_{[i],P}'=\Delta_{[i+1],P}, i=1,\ldots,m-1,\\
&\Delta_{[m],P}'=t-S.
\end{split}
\end{equation}
In the scheduling strategy $\pi$, this packet can be generated from any source. Thus, for all cases of strategy $\pi$, it must hold that
\begin{equation}\label{prlem2thm1eq2}
\Delta_{[i],\pi}'\geq\Delta_{[i+1],\pi}, i=1,\ldots,m-1.
\end{equation}
By combining \eqref{lem2thm1eq1}, \eqref{prlem2thm1eq1}, and \eqref{prlem2thm1eq2}, we have
\begin{equation}\label{prlem2thm1eq3}
\Delta_{[i],\pi}'\geq\Delta_{[i+1],\pi}\geq\Delta_{[i+1],P}=\Delta_{[i],P}', i=1,\ldots,m-1.
\end{equation}
In addition, since the same packet is also delivered under the scheduling strategy $\pi$, the source from which this packet is generated under policy $\pi$ will have the minimum age after the delivery, i.e., we have
\begin{equation}\label{prlem2thm1eq4}
\Delta_{[m],\pi}'=t-S=\Delta_{[m],P}'.
\end{equation}
By this, \eqref{lem2thm1eq2} is proven. 
\end{proof}
\begin{proof}[Proof of Lemma \ref{lem1thm1}]
Using the coupling between the system state processes, and for any given sample path of the packet service times, we consider two cases:

\emph{Case 1:} When there is no packet delivery, the age of each source grows linearly with a slope 1.

\emph{Case 2:} When a packet is delivered, the ages of the sources evolve according to Lemma \ref{lem2thm1}.

By induction over time, we obtain
\begin{equation}
\Delta_{[i],P}(t)\leq \Delta_{[i],\pi}(t), i=1,\ldots,m, t\geq 0.
\end{equation}
Hence, \eqref{Thm1eq1} follows which implies \eqref{lem1thm1eq1} by Theorem 6.B.30 of \cite{shaked2007stochastic}. This completes the proof.
\end{proof}
\begin{proof}[Proof of Proposition \ref{Thm1}]
Since the TaPA and TaA for any scheduling policy $\pi$ are the expectation of non-decreasing functionals of the process $\{\mathbf{\Delta}_\pi(t), t\geq 0\}$, \eqref{lem1thm1eq1} implies \eqref{Thm1eq1b} and \eqref{Thm1eq2b} using the properties of stochastic ordering \cite{shaked2007stochastic}. This completes the proof. 
\end{proof}

\ifreport 
\section{Proof of Lemma \ref{peak_zero_wait_opt}}\label{Appendix_A'}
We prove Proposition \ref{peak_zero_wait_opt} by showing that the TaPA is an increasing function of the packets waiting times $Z_i$'s. For notational simplicity, let $a_{li}$ denote the age value for the source $l$ at time $D_i$, i.e., $a_{li}=\Delta_l(D_i)$. Since the age process increases linearly with time when there is no packet delivery, we have
\begin{equation}
\Delta_{r_i}(D_i^{-})=a_{r_i(i-1)}+Z_{i-1}+Y_{i},
\end{equation}
where $a_{r_i(i-1)}=\Delta_{r_i}(D_{i-1})$.
Substituting by this in \eqref{peak_age_def}, we get
\begin{equation}\label{peak_age_inc_eq}
\!\!\!\!\!\Delta_{\text{peak}}(\pi_{\text{MAF}},f)=\limsup_{n\rightarrow\infty}\frac{1}{n}\mathbb{E}\left[\sum_{i=1}^{n}a_{r_i(i-1)}\!+\!Z_{i-1}\!+\!Y_{i}\right].
\end{equation}
Since we follow the MAF scheduler, the last serves of the source $r_i$ before time $D_{i-1}$ occurs at time $D_{i-m}$. Since the age process increases linearly if there is no packet delivery, we have
\begin{equation}
a_{r_i(i-1)}=D_{i-1}-D_{i-m}+Y_{i-m},
\end{equation}
where $Y_{i-m}$ is the age value of the source $r_i$ at time $D_{i-m}$, i.e., $\Delta_{r_i}(D_{i-m})=Y_{i-m}$. Note that $D_{i-1}=Y_{i-1}+Z_{i-2}+D_{i-2}$. Repeating this, we can express $(D_{i-1}-D_{i-m})$ in terms of $Z_i$'s and $Y_i$'s, and hence we get
\begin{equation}\label{ari_interms_ziandyi}
a_{r_i(i-1)}=\sum_{k=1}^mY_{i-k}+\sum_{k=2}^{m}Z_{i-k}.
\end{equation}
For example, in Fig. \ref{fig:ages_evolv}, we have $a_{22}=Y_1+Z_1+Y_2$. Note that the regularity of \eqref{ari_interms_ziandyi} occurs after the first $m$ transmissions (i.e., \eqref{ari_interms_ziandyi}  is valid after the first $m$ transmissions). For simplicity, we can omit the first $m$ peaks in \eqref{peak_age_inc_eq} ($\sum_{i=1}^{m}\Delta_{r_i}(D_i^{-})$), as $\lim_{n\rightarrow \infty}\frac{\sum_{i=1}^{m}\Delta_{r_i}(D_i^{-})}{n}= 0$ (observe that $\mathcal{Z}$ and $\mathcal{Y}$ are bounded), and this will not affect the value of the TaPA. This with substituting by \eqref{ari_interms_ziandyi} in \eqref{peak_age_inc_eq}, we get
\begin{equation}\label{peak_age_inc_eq1}
\!\!\Delta_{\text{peak}}(\pi_{\text{MAF}},f)=\limsup_{n\rightarrow\infty}\frac{1}{n}\mathbb{E}\left[\sum_{i=m+1}^{n}\left(\sum_{k=0}^mY_{i-k}+\sum_{k=1}^{m}Z_{i-k}\right)\!\right].\!\!
\end{equation}
From \eqref{peak_age_inc_eq1}, it follows that the TaPA is an increasing function of the waiting times. This completes the proof. 

\section{Proof of Lemma \ref{pro_simple}}\label{Appendix_B}
Part (i) is proven in two steps:

\emph{Step 1:} We will prove that $\bar{\Delta}_{\text{avg-opt}}\leq \beta$ if and only if $p(\beta)\leq 0$.

If $\bar{\Delta}_{\text{avg-opt}}\leq \beta$, there exists a sampling policy $f=(Z_0,Z_1,\ldots)\in\mathcal{F}$ that is feasible for \eqref{optimal_eq_sampler} and \eqref{equivilent_optimal_sampler}, which satisfies
\begin{equation}\label{lem3_simp_eq1}
\limsup_{n\rightarrow\infty}\frac{\sum_{i=0}^{n-1}\mathbb{E}\left[A_{i}(Z_i+Y_{i+1})+\frac{m}{2}(Z_i+Y_{i+1})^2\right]}{\sum_{i=0}^{n-1}\mathbb{E}[Z_i+Y_{i+1}]}\leq \beta.
\end{equation}
Hence,
\begin{equation}\label{lem3_simp_eq2}
\!\!\!\!\!\!\!\!\!\!\limsup_{n\rightarrow\infty}\frac{\frac{1}{n}\sum_{i=0}^{n-1}\mathbb{E}\left[(A_{i}-\beta)(Z_i+Y_{i+1})+\frac{m}{2}(Z_i+Y_{i+1})^2\right]}{\frac{1}{n}\sum_{i=0}^{n-1}\mathbb{E}[Z_i+Y_{i+1}]}\leq 0.
\end{equation} 
Since $Z_i$'s and $Y_i$'s are bounded and positive and $\mathbb{E}[Y_i]>0$ for all\footnote{The boundedness of $Z_i$ for each $i$ does not only follow from the upper bound $M$, but also from the fact that the zero-wait sampler will generate a lower value of the TaA if $Z_i=\infty$ for some $i$. Hence, $Z_i$ must be bounded for all $i$.} $i$, we have $0<\liminf_{n\rightarrow\infty}$ $\frac{1}{n}\sum_{i=0}^{n-1}\mathbb{E}[Z_i+Y_{i+1}]\leq \limsup_{n\rightarrow\infty}$ $\frac{1}{n}\sum_{i=0}^{n-1}\mathbb{E}[Z_i+Y_{i+1}]\leq q$ for some $q\in\mathbb{R}^{+}$. By this, we get
\begin{equation}\label{lem3_simp_eq3}
\limsup_{n\rightarrow\infty}\frac{1}{n}\sum_{i=0}^{n-1}\mathbb{E}\left[(A_{i}-\beta)(Z_i+Y_{i+1})+\frac{m}{2}(Z_i+Y_{i+1})^2\right]\leq 0.
\end{equation}
Therefore, $p(\beta)\leq 0$.

In the reverse direction, if $p(\beta)\leq 0$, then  there exists a sampling policy $f=(Z_0,Z_1,\ldots)\in\mathcal{F}$ that is feasible for \eqref{optimal_eq_sampler} and \eqref{equivilent_optimal_sampler}, which satisfies \eqref{lem3_simp_eq3}. Since we have $0<\liminf_{n\rightarrow\infty}$ $\frac{1}{n}\sum_{i=0}^{n-1}\mathbb{E}[Z_i+Y_{i+1}]\leq\limsup_{n\rightarrow\infty}$ $\frac{1}{n}\sum_{i=0}^{n-1}\mathbb{E}[Z_i+Y_{i+1}]\leq q$, we can divide \eqref{lem3_simp_eq3} by $\liminf_{n\rightarrow\infty}\frac{1}{n}$ $\sum_{i=0}^{n-1}\mathbb{E}[Z_i+Y_{i+1}]$ to get \eqref{lem3_simp_eq2}, which implies \eqref{lem3_simp_eq1}. Hence, $\bar{\Delta}_{\text{avg-opt}}\leq \beta$. By this, we have proven that  $\bar{\Delta}_{\text{avg-opt}}\leq \beta$ if and only if $p(\beta)\leq 0$.

\emph{Step 2:} We need to prove that $\bar{\Delta}_{\text{avg-opt}}< \beta$ if and only if $p(\beta)< 0$. This statement can be proven by using the arguments in Step 1, in which ``$\leq$" should be replaced by ``$<$''. Finally, from the statement of Step 1, it immediately follows that $\bar{\Delta}_{\text{avg-opt}}> \beta$ if and only if $p(\beta)> 0$. This completes part (i).

Part(ii): We first show that each optimal solution to \eqref{optimal_eq_sampler} is an optimal solution to \eqref{equivilent_optimal_sampler}. 
By the claim of part (i), $p(\beta)=0$ is equivalent to $\bar{\Delta}_{\text{avg-opt}}=\beta$. Suppose that policy $f=(Z_0,Z_1,\ldots)\in\mathcal{F}$ is an optimal solution to \eqref{optimal_eq_sampler}. Then, $\Delta_{\text{avg}(\pi_{\text{MAF}},f)}=\bar{\Delta}_{\text{avg-opt}}=\beta$. Applying this in the arguments of \eqref{lem3_simp_eq1}-\eqref{lem3_simp_eq3}, we can show that policy $f$ satisfies
\begin{equation}
\limsup_{n\rightarrow\infty}\frac{1}{n}\sum_{i=0}^{n-1}\mathbb{E}\left[(A_{i}-\beta)(Z_i+Y_{i+1})+\frac{m}{2}(Z_i+Y_{i+1})^2\right]= 0.
\end{equation}
This and $p(\beta)=0$ imply that policy $f$ is an optimal solution to \eqref{equivilent_optimal_sampler}.

Similarly, we can prove that each optimal solution to \eqref{equivilent_optimal_sampler} is an optimal solution to \eqref{optimal_eq_sampler}. By this, part (ii) is proven.

\section{Proof of Proposition \ref{thm2}}\label{Appendix_C}
According to \cite[Proposition 4.2.1 and Proposition 4.2.6]{Bertsekas1996bookDPVol2}, it is enough to show that for every two states $\mathbf{s}$ and $\mathbf{s}'$, there exists a stationary deterministic policy $f$ such that for some $k$, we have
\begin{equation}\label{cond_eq_req1}
\mathbb{P}\left[\mathbf{s}(k)=\mathbf{s}'\vert\mathbf{s}(0)=\mathbf{s}, f\right] > 0.
\end{equation}
From the state evolution equation \eqref{state_evol}, we can observe that any state in $\mathcal{S}$ can be represented in terms of the packet waiting and service times. This implies \eqref{cond_eq_req1}. To make it clearer, we consider the following case of the 3 sources system:

Assume that the elements of state $\mathbf{s}'$ are as follows
\begin{equation}
\begin{split}
&a_{[1]}'=y_3+z_2+y_2+z_1+y_1,\\
&a_{[2]}'=y_3+z_2+y_2,\\
&a_{[3]}'=y_3,
\end{split}
\end{equation}
where $y_i$'s and $z_i$'s are any arbitrary elements in $\mathcal{Y}$ and $\mathcal{Z}$, respectively. Then, we will show that from any arbitrary state $\mathbf{s}=(a_{[1]},a_{[2]},a_{[3]})$, a sequence of service and waiting times can be followed to reach state $\mathbf{s}'$. If we have $Z_0=z_1$, $Y_1=y_1$, $Z_1=z_1$, $Y_2=y_2$, $Z_2=z_2$, and $Y_3=y_3$, then according to \eqref{state_evol}, we have in the first stage
\begin{equation}
\begin{split}
&a_{[1]1}=a_{[2]}+z_1+y_1,\\
&a_{[2]1}=a_{[3]}+z_1+y_1,\\
&a_{[3]1}=y_1,
\end{split}
\end{equation}
and in the second stage, we have
\begin{equation}
\begin{split}
&a_{[1]2}=a_{[3]}+z_1+y_2+z_1+y_1,\\
&a_{[2]2}=y_2+z_1+y_1,\\
&a_{[3]2}=y_2,
\end{split}
\end{equation}
and in the third stage, we have
\begin{equation}
\begin{split}
&a_{[1]3}=y_3+z_2+y_2+z_1+y_1=a_{[1]}',\\
&a_{[2]3}=y_3+z_2+y_2=a_{[2]}',\\
&a_{[3]3}=y_3=a_{[3]}'.
\end{split}
\end{equation}
Hence, a stationary deterministic policy $f$ can be designed to reach state $\mathbf{s}'$ from state $\mathbf{s}$ in 3 stages, if the aforementioned sequence of service times occurs. This implies that
\begin{equation}
\mathbb{P}\left[\mathbf{s}(3)=\mathbf{s}'\vert\mathbf{s}(0)=\mathbf{s}, f\right] =\prod_{i=1}^3 \mathbb{P}(Y_i=y_i)>0,
\end{equation}
where we have used that $Y_i$'s are \emph{i.i.d.}\footnote{We assume that all elements in $\mathcal{Y}$ have a strictly positive probability, where the elements with zero probability can be removed without affecting the proof.} The previous argument can be generalized for any $m$ sources system. In particular, a forward induction over $m$ can be used to show the result, as \eqref{cond_eq_req1} trivially holds for $m=1$, and the previous argument can be used to show that \eqref{cond_eq_req1} holds for any general $m$. This completes the proof.

\section{Proof of Proposition \ref{th_thm}}\label{Appendix_E}
We prove Proposition \ref{th_thm} into two steps:

\textbf{Step 1}: We first address an infinite horizon discounted cost problem. Then, we connect it to the average cost per stage problem. In particular, we show that $J_\alpha(\mathbf{s})$ is non-decreasing in $\mathbf{s}$, which together with \eqref{relative_cost_eq} imply that $h(\mathbf{s})$ is non-decreasing in $\mathbf{s}$ as well.

Given an initial state $\mathbf{s}(0)$, the total expected discounted cost under a sampling policy $f\in\mathcal{F}$ is given by
\begin{equation}
J_\alpha(\mathbf{s}(0); f)=\limsup_{n\rightarrow\infty}\mathbb{E}\left[ \sum_{i=0}^{n-1}\alpha^iC(\mathbf{s}(i), Z_i)\right],
\end{equation}
where $0<\alpha<1$ is the discount factor. The optimal total expected $\alpha$-discounted cost function is defined by
\begin{equation}
J_\alpha(\mathbf{s})=\min_{f\in\mathcal{F}}J_\alpha(\mathbf{s}; f),~\mathbf{s}\in\mathcal{S}.
\end{equation}
A policy is said to be $\alpha$-optimal if it minimizes the total expected $\alpha$-discounted cost. The discounted cost optimality equation of $J_\alpha(\mathbf{s})$ is discussed below.
\begin{proposition}
The optimal total expected $\alpha$-discounted cost $J_\alpha(\mathbf{s})$ satisfies
\begin{equation}\label{optimality_cond}
J_\alpha(\mathbf{s})=\min_{z\in\mathcal{Z}} C(\mathbf{s}, z)+\alpha\sum_{\mathbf{s}'\in\mathcal{S}}\mathbb{P}_{\mathbf{s}\mathbf{s}'}(z)J_\alpha(\mathbf{s}').
\end{equation}
Moreover, a stationary deterministic policy that attains the minimum in equation \eqref{optimality_cond} for each $\mathbf{s}\in\mathcal{S}$ will be an $\alpha$-optimal policy. Also, let $J_{\alpha,0}(\mathbf{s})=0$ for all $\mathbf{s}$ and for any $n\geq 0$,
\begin{equation}\label{value_iter_disc}
J_{\alpha,n+1}(\mathbf{s})=\min_{z\in\mathcal{Z}} C(\mathbf{s}, z)+\alpha\sum_{\mathbf{s}'\in\mathcal{S}}\mathbb{P}_{\mathbf{s}\mathbf{s}'}(z)J_{\alpha,n}(\mathbf{s}').
\end{equation}
Then, we have $J_{\alpha,n}(\mathbf{s})\rightarrow J_{\alpha}(\mathbf{s})$ as $n\rightarrow\infty$ for every $\mathbf{s}$, and $\alpha$.
\end{proposition}
\begin{proof}
Since we have bounded cost per stage, the proposition follows directly from
\cite[Proposition 1.2.2 and Proposition 1.2.3]{Bertsekas1996bookDPVol2}, and \cite{sennott1989average}.
\end{proof}
Next, we use the optimality equation \eqref{optimality_cond} and the value iteration in \eqref{value_iter_disc} to prove that $J_\alpha(\mathbf{s})$ is non-decreasing in $\mathbf{s}$.
\begin{lemma}\label{lem2}
The optimal total expected $\alpha$-discounted cost function $J_\alpha(\mathbf{s})$ is non-decreasing in $\mathbf{s}$.
\end{lemma}
\begin{proof}
We use induction on $n$ in equation \eqref{value_iter_disc} to prove Lemma \ref{lem2}. Obviously, the result holds for $J_{\alpha,0}(\mathbf{s})$.

Now, assume that $J_{\alpha,n}(\mathbf{s})$ is non-decreasing in $\mathbf{s}$. We need to show that for any two states $\mathbf{s}_1$ and $\mathbf{s}_2$ with $\mathbf{s}_1\leq \mathbf{s}_2$, we have $J_{\alpha,n+1}(\mathbf{s}_1)\leq J_{\alpha,n+1}(\mathbf{s}_2)$. First, we note that the expected cost per stage $C(\mathbf{s},z)$ is non-decreasing in $\mathbf{s}$, i.e., we have
\begin{equation}\label{pf1}
C(\mathbf{s}_1,z)\leq C(\mathbf{s}_2,z).
\end{equation}
From the state evolution equation \eqref{state_evol} and the transition probability equation \eqref{trans_prob_eq}, the second term of the right-hand side (RHS) of \eqref{value_iter_disc} can be rewritten as
\begin{equation}
\sum_{\mathbf{s}'\in\mathcal{S}}\mathbb{P}_{\mathbf{s}\mathbf{s}'}(z)J_{\alpha,n}(\mathbf{s}')=\sum_{y\in\mathcal{Y}}\mathbb{P}(Y=y)J_{\alpha,n}(\mathbf{s}'(z,y)),
\end{equation}
where $\mathbf{s}'(z,y)$ is the next state from state $\mathbf{s}$ given the values of $z$ and $y$. Also, according to the state evolution equation \eqref{state_evol}, if the next states of $\mathbf{s}_1$ and $\mathbf{s}_2$ for given values of $z$ and $y$ are $\mathbf{s}'_1(z,y)$ and $\mathbf{s}'_2(z,y)$, respectively, then we have $\mathbf{s}'_1(z,y)\leq \mathbf{s}'_2(z,y)$. This implies that
\begin{equation}\label{pf2}
\sum_{y\in\mathcal{Y}}\mathbb{P}(Y=y)J_{\alpha,n}(\mathbf{s}'_1(z,y))\leq \sum_{y\in\mathcal{Y}}\mathbb{P}(Y=y)J_{\alpha,n}(\mathbf{s}'_2(z,y)),
\end{equation}
where we have used the induction assumption that  $J_{\alpha,n}(\mathbf{s})$ is non-decreasing in $\mathbf{s}$. Using \eqref{pf1}, \eqref{pf2}, and that fact that the minimum operator in \eqref{value_iter_disc} holds the non-decreasing property, we conclude that 
\begin{equation}
J_{\alpha,n+1}(\mathbf{s}_1)\leq J_{\alpha,n+1}(\mathbf{s}_2).
\end{equation}
This completes the proof.
\end{proof}

\textbf{Step 2:} We use Step 1 to prove Proposition \ref{th_thm}. From Step 1, we have that $h(\mathbf{s})$ is non-decreasing in $\mathbf{s}$. Similar to Step 1, this implies that the second term of the right-hand side (RHS) of \eqref{bell1'} ($\sum_{\mathbf{s}'\in\mathcal{S}}\mathbb{P}_{\mathbf{s}\mathbf{s}'}(z)h(\mathbf{s}')$) is non-decreasing in $\mathbf{s}'$. Moreover, from the state evolution \eqref{state_evol}, we can notice that, for any state $\mathbf{s}$, the next state $\mathbf{s}'$ is increasing in $z$. This argument implies that the second term of the right-hand side (RHS) of \eqref{bell1'} ($\sum_{\mathbf{s}'\in\mathcal{S}}\mathbb{P}_{\mathbf{s}\mathbf{s}'}(z)h(\mathbf{s}')$) is increasing in $z$. Thus, the value of $z\in\mathcal{Z}$ that achieves the minimum value of this term is zero. If, for a given state $\mathbf{s}$,
the value of $z\in\mathcal{Z}$ that achieves the minimum value of the cost function $C(\mathbf{s},z)$ is zero, then $z=0$ solves the RHS of \eqref{bell1'}. From \eqref{cost}, we observe that $C(\mathbf{s},z)$ is convex in $z$. Also, the value of $z$ that minimizes $C(\mathbf{s},z)$ is $ \frac{\bar{\Delta}_{\text{avg-opt}}-A_{s}-m\mathbb{E}[Y]}{m}$. This implies that for any state $\mathbf{s}$ with $A_{s}\geq (\bar{\Delta}_{\text{avg-opt}}-m\mathbb{E}[Y])$, $z=0$ minimizes $C(\mathbf{s},z)$ in the domain $\mathcal{Z}$. Hence, for any state $\mathbf{s}$ with $A_{s}\geq (\bar{\Delta}_{\text{avg-opt}}-m\mathbb{E}[Y])$, $z=0$ solves the RHS of \eqref{bell1'}. This completes the proof.

\section{Proof of Theorem \ref{upperbound_removal}}\label{Appendix_D}
We prove Theorem \ref{upperbound_removal} into 2 steps.

\textbf{Step 1:} We show that when the upper bound on the waiting times keeps increasing, the optimal waiting times are still bounded. Let $z^*(\mathbf{s})$ be the optimal waiting time for the state $\mathbf{s}$. Observe that the TaA of the zero-wait sampler, denoted by $\bar{\Delta}_0$, must be greater than the optimum TaA. Thus, from Proposition \ref{th_thm}, we conclude that
when $A_s\geq \bar{\Delta}_0$ (recall that $A_\mathbf{s}$ is the sum of ages of state $\mathbf{s}$), the optimal waiting time $z^*(\mathbf{s})$ is zero (this also can be deduced from Problem \eqref{equivilent_optimal_sampler}, where we can follow a similar argument used in Proposition \ref{th_thm} to show that the optimal value of $Z_i$ is zero whenever $A_i\geq \bar{\Delta}_0$). Hence, we can restrict our focus on the age values whose summations are less than $\bar{\Delta}_0$ (which is a finite subset of the system state space). It is obvious that the optimal waiting time at any stage cannot increase without limit (i.e., goes to  $\infty$) as the zero-wait sampler can provide a lower TaA in this case. Thus, the optimal waiting times, when $M$ is large enough, are upper bounded by the following bound.
\begin{equation}\label{upper_bound_eq_B}
B=\max\{z^*(\mathbf{s}): A_\mathbf{s}<\bar{\Delta}_0\}.
\end{equation}
 Observe that there is a limit after which increasing $M$ just adds states with $A_s\geq\bar{\Delta}_0$, and hence does not affect the bound in \eqref{upper_bound_eq_B}.

\textbf{Step 2:} 
From Step 1,
 we can conclude that there exists $N_o\geq B$ such that we have $f^*_{M}=f^*_{N}$ for any $M, N\geq N_o$. This implies that $f^*_{\infty}=f^*_{M}$ for any $M\geq N_o$. Moreover, substituting by this in Problem \eqref{optimal_eq_sampler}, we get $\bar{\Delta}^{\infty}_{\text{avg-opt}}=\bar{\Delta}^{M}_{\text{avg-opt}}$  for any $M\geq N_o$. These prove \eqref{opt_samp_inf}, which completes the proof. 

\fi